\documentclass[conference,letter]{IEEEtran}

\usepackage[utf8]{inputenc} 
\usepackage[T1]{fontenc}
\usepackage{url}              
\usepackage{cite}             
\usepackage{amsmath,amsfonts}
\usepackage{graphicx}
\usepackage{algorithmic}
\usepackage{array}
\usepackage[caption=false,font=normalsize,labelfont=normalsize,textfont=normalsize]{subfig}
\usepackage{amsmath,amssymb,amsfonts}
\usepackage{algorithmic}
\usepackage{graphicx}
\usepackage{multirow}
\usepackage{amsmath}
\usepackage{amssymb}
\usepackage{aligned-overset}
\usepackage{latexsym}
\usepackage{epsfig}
\usepackage{epsf}
\usepackage{graphics}
\usepackage{listings}
\usepackage{color}
\usepackage{epstopdf}
\usepackage{multirow}
\usepackage{psfrag}
\usepackage{fancyhdr}
\usepackage{multirow}
\usepackage{textcomp}
\usepackage{float}
\usepackage{subfig}
\usepackage{xcolor}
\usepackage{caption}
\usepackage{times}  
\usepackage{helvet}  
\usepackage{courier}  
\usepackage{graphicx} 
\usepackage{caption} 
\usepackage{cite}
\usepackage{algorithm}
\usepackage{algorithmic}
\usepackage{balance}
\usepackage{xurl}
\usepackage{multicol}
\usepackage{graphicx}
\usepackage{listings}
\usepackage{multicol}
\usepackage[utf8]{inputenc} 
\usepackage[T1]{fontenc}    
\usepackage{url}            
\usepackage{booktabs}       
\usepackage{amsfonts}       
\usepackage{nicefrac}       
\usepackage{microtype}      
\usepackage{xcolor}         
\usepackage{mathtools, amsmath, amsthm}
\usepackage{tikz,pgfplots}
\usepackage{newfloat}
\usepackage{listings}
\usepackage{multicol}

\newtheorem{theorem}{{Theorem}}
\newtheorem{lemma}[theorem]{{Lemma}}

\theoremstyle{definition}
\newtheorem{definition}{Definition}
\newtheorem{remark}{{\textbf{Remark}}}

\usepackage{textcomp}
\usepackage{stfloats}
\usepackage{url}
\usepackage{verbatim}
\usepackage{graphicx}
\usepackage{mleftright}       
\mleftright                   

\usepackage{graphicx}         
\usepackage{booktabs}         

\newcommand{\prot}{$\mathrm{PROT^{PP}}$-$\mathrm{S}$-$\mathrm{Hist_{PP}}$ }
\newcommand{\protns}{$\mathrm{PROT^{PP}}$-$\mathrm{S}$-$\mathrm{Hist_{PP}}$}
\def\BibTeX{{\rm B\kern-.05em{\sc i\kern-.025em b}\kern-.08em
    T\kern-.1667em\lower.7ex\hbox{E}\kern-.125emX}}
\usepackage{balance}
\begin{document}

\IEEEoverridecommandlockouts
\title{A Locally Differential Private Coding-Assisted Succinct Histogram Protocol
}

\author{%
  \IEEEauthorblockN{Hsuan-Po Liu\IEEEauthorrefmark{1} and Hessam Mahdavifar\IEEEauthorrefmark{1}\IEEEauthorrefmark{2}}
  \IEEEauthorblockA{\IEEEauthorrefmark{1}Department of Electrical Engineering and Computer Science, University of Michigan, Ann Arbor, MI 48109, USA}
  \IEEEauthorblockA{\IEEEauthorrefmark{2}Department of
Electrical and Computer Engineering, Northeastern University, Boston, MA
02115, USA}
  \IEEEauthorblockA{Emails: hsuanpo@umich.edu, h.mahdavifar@northeastern.edu}
}

\maketitle

\begin{abstract}

A succinct histogram captures frequent items and their frequencies across clients and has become increasingly important for large-scale, privacy-sensitive machine learning applications. To develop a rigorous framework to guarantee privacy for the succinct histogram problem, local differential privacy (LDP) has been utilized and shown promising results. 
To preserve data utility under LDP, which essentially works by intentionally adding noise to data, error-correcting codes naturally emerge as a promising tool for reliable information collection. 
This work presents the first practical $(\epsilon,\delta)$-LDP protocol for constructing succinct histograms using error-correcting codes. To this end, polar codes and their successive-cancellation list (SCL) decoding algorithms are leveraged as the underlying coding scheme. More specifically, our protocol introduces Gaussian-based perturbations to enable efficient soft decoding. Experiments demonstrate that our approach outperforms prior methods, particularly for items with low true frequencies, while maintaining similar frequency estimation accuracy.
\end{abstract}

\begin{IEEEkeywords}
Succinct histogram, local differential privacy, error-correcting codes.
\end{IEEEkeywords}

\section{Introduction}
\label{sec:intro}
In the era of big data, collecting large-scale contextual information is key to generating accurate statistics, training sophisticated machine learning models, and turning raw data into actionable insights, which has been well-studied for decades \cite{leskovec2020mining}. 
However, the state-of-the-art protocols that directly gather user details raise serious privacy concerns. Data on local devices often includes sensitive personal information, such as behavioral traits, usage patterns, or private preferences, which raises substantial privacy concerns when raw records are collected without protection. By adopting privacy-preserving techniques \cite{dwork2006calibrating}, such as perturbations on the information shared by the devices, we can enable powerful analytics and learning workflows that, otherwise, would be impractical or impossible under traditional data collection paradigms.

Differential Privacy (DP) \cite{dwork2006calibrating, dwork2014algorithmic} provides a mathematically rigorous framework for privacy-preserving data analysis. It ensures that the addition or removal of any single individual’s record has a negligible effect on the output of an algorithm, thereby limiting the risk of sensitive information leakage. 
In the centralized DP model, a trusted curator aggregates raw data and applies this noise before releasing outputs. However, such trustworthiness assumptions may not hold in many practical settings, simply due to the large scale of the systems and the volume of data that are collected from a large number of clients \cite{liu2023differentially}. Local DP (LDP) \cite{ duchi2013local}, on the other hand, eliminates the need for a trusted intermediary by shifting the perturbation step to each user’s device: individuals perturb their data locally before transmission.
Therefore, the LDP framework better meets the needs of modern large-scale data collection.

A \emph{succinct histogram} lists heavy hitters and estimates their frequencies \cite{bassily2015local}. Coding-based methods \cite{bassily2015local, bassily2020practical,bun2019heavy} encode client data with error-correcting codes, apply random perturbations to satisfy $\epsilon$-LDP, and decode at the server to identify heavy hitters. Early work \cite{bassily2015local} introduced a protocol, referred to as \prot, to solve the \emph{unique heavy hitter} problem, aiming to identify items held by at least an $\eta$ fraction of clients. Although \prot provided error bounds, it lacked practical code design considerations and operated on hard decisions. 
Later refinements in \cite{bassily2020practical,bun2019heavy} improved the theoretical error bounds, but remained untested in practice.

This paper proposes the first practical $(\epsilon,\delta)$-LDP protocol for succinct histograms that utilizes error-correcting codes, in particular, polar codes \cite{arikan2009channel} and their successive-cancellation list (SCL) decoder \cite{tal2015list}. By adopting the analytic Gaussian mechanism \cite{balle2018improving} for perturbation, the protocol enables efficient soft decoding. This is appealing from a practical perspective, indicating that existing mechanisms for reliable communication currently deployed, e.g., 5G protocols, can be re-used to guarantee privacy as well without incurring additional computational cost \cite{liu2024projective}. 
We provide theoretical bounds on the frequency estimation error and extend the design to the general succinct histogram setting \cite{bassily2015local}.
For comparison, we implement \prot with polar codes and maximum likelihood decoding. 
Experiments show that our protocol achieves a lower threshold on $\eta$, offering greater robustness in decoding errors while maintaining a comparable frequency estimation accuracy. The main contributions are: (i) the first practical $(\epsilon,\delta)$-LDP coding-assisted succinct histogram protocol; (ii) leveraging analytic Gaussian mechanism to enable soft decoding; and (iii) experimental validation demonstrating improved robustness over prior work.




The rest of the paper is structured as follows.
In Section~\ref{sec:prelim}, some preliminaries are provided. 
In Section~\ref{sec:protocol}, we present the proposed protocol with $(\epsilon,\delta)$-LDP guarantees, followed by the analysis of frequency estimation error.
The simulation results are included in Section~\ref{sec:sim}.
Finally, we conclude the paper in Section~\ref{sec:con}.

\section{Preliminaries}
\label{sec:prelim}
\subsection{Succinct histogram}
The succinct histogram problems analyze the most frequent items, the heavy hitters, from clients exceeding a certain frequency threshold.
We define a succinct histogram and the unique heavy hitter problem we consider in this paper. The details can be found in \cite{bassily2015local}.
\begin{definition}[Succinct Histogram]
A succinct histogram is a data structure that provides a list of frequent items called heavy hitters, together with their corresponding estimated frequencies among the clients.

\end{definition}
\begin{definition}[The unique heavy hitter problem]
\label{def:unique}
At least $\eta$ fraction of clients hold the same item as an $k$-bit binary vector $\mathbf{u}^*$ for some $\mathbf{u}^*\in\mathcal{U}=\{0,1\}^{k}$ unknown to the untrusted server, while other clients hold a special symbol $\bot$, to represent "no item."
The goal is to design an efficient $(\epsilon,\delta)$-LDP protocol of a succinct histogram, for $\eta$ as small as possible.
\end{definition}

\subsection{Local Differential Privacy}
The main idea to achieve $(\epsilon,\delta)$-LDP is to intentionally add perturbation via random noise sampled from a probability distribution at the client's side, before sharing information with the untrusted server. The definition of $(\epsilon,\delta)$-LDP is provided below. 
The details can be found in \cite{dwork2014algorithmic,duchi2013local}.
\begin{definition}[$(\epsilon,\delta)$-LDP]
\label{def:LDP}
Let $u$ and $u^\prime$ be neighboring datasets that differ only by a single event, i.e., $\mathrm{dist}(u,u^\prime)=1$. 
The neighboring datasets $u$ and $u^\prime$ satisfy $(\epsilon,\delta)$-LDP for any $\epsilon>0$ under a randomized mechanism $\mathcal{R}$ that under $\mathcal{E} \subseteq \mathrm{Range}(\mathcal{R})$,
\begin{equation}
\label{eq:LDP}
\mathbb{P}[\mathcal{R}(u)\in \mathcal{E}]\leq e^\epsilon\cdot\mathbb{P}[\mathcal{R}(u^\prime)\in \mathcal{E}]+\delta,
\end{equation}
where $\delta$ represents the failure probability.
\end{definition}

The sensitivity for a query function $q(\cdot)$ is the maximum difference between two queries.
\begin{definition}[Sensitivity]
\label{def:sen}
For two neighboring datasets $u$ and $u^\prime$ in an individual client together with a query function $q:\mathcal{D}\to\mathbb{R}$, the sensitivity is defined as follows:
\begin{equation}
\Delta \overset{\mathrm{def}}{=} \max_{\mathrm{dist}(u,u^\prime)=1}||q(u)-q(u^\prime)||_2.
\end{equation}
\end{definition}
The Gaussian mechanism is also defined as follows.
\begin{definition}[Gaussian mechanism]
\label{eq:gm}
Consider applying a query function $q$ on a dataset $u$. Then, the Gaussian mechanism $\mathcal{R}$ is defined as 
\begin{equation}
\mathcal{R}(u) \overset{\mathrm{def}}{=} q(u)+\mathcal{N}(0,\sigma^2),
\end{equation}
which adds random noise to the query result according to a zero-mean Gaussian distribution with variance $\sigma^2$. Note that $\sigma$ should be a function of $\epsilon,\delta$, and $\Delta$.
\end{definition}
Intuitively, $\sigma$ should be proportional to sensitivity $\Delta$. That is, a greater value of sensitivity for a dataset implies that we need a larger perturbation of the query, making it more indistinguishable from the adversary. We indicate the expression of $\sigma$ for the randomized mechanism in our proposed protocol in the next section.

\subsection{Polar Codes}
\paragraph{Coding structure}
An $(n,k)$-Polar code \cite{arikan2009channel} is a linear block code that has a code dimension of $k$ and the code length as $n=2^m$, where $m$ is the Kronecker power such that the generator matrix $\mathbf{G}$ follows the transformation $\mathbf{G}=\mathbf{G}_2^{\otimes m}$ given the base matrix $\mathbf{G}_2=\begin{bmatrix}
    1 & 0\\1 &1
\end{bmatrix}$.
We denote the binary information vector by $\mathbf{u}$ and the transmitted vector as a codeword corresponding to each information vector by $\mathbf{x}$, such that $\mathbf{x=uG}$, where $\mathbf{x}\in\{0,1\}^n$ and $\mathbf{u}\in\{0,1\}^k$.
The key idea behind polar codes is the phenomenon of channel polarization, where certain bit-channels become highly reliable while others become completely noisy as the code length increases.
A polar encoder exploits this polarization effect to transmit information bits through reliable bit-channels, while fixing the values of the less reliable ones to zeros, called frozen bits.

\paragraph{Decoder}
The SC decoder \cite{arikan2009channel} for polar codes operates by sequentially estimating each bit in the codeword using a recursive computation of log-likelihood ratios (LLRs).
To estimate the $j$th message bit ${u}_j$ in $\mathbf{u}$, the SC decoder makes the following decision based on the received vector $\mathbf{y}$ and past decisions $\hat{\mathbf{u}}_{1}^{j-1}=[\hat{u}_1,\dots,\hat{u}_{j-1}]$ as
\begin{equation}
\hat{u}_i = 
\begin{cases}
0 & \text{if } L^{(j)}\left( \mathbf{y}, \hat{\mathbf{u}}_{1}^{j-1} \right) \geq 1 \\
1 & \text{otherwise}
\end{cases}
\end{equation}
where $L^{(j)}\left( \mathbf{y}, \hat{\mathbf{u}}_{1}^{j-1} \right)  \overset{\mathrm{def}}{=} 
\frac{W^{(j)}\left( \mathbf{y}, \hat{\mathbf{u}}_{1}^{j-1} \mid 0 \right)}
{W^{(j)}\left( \mathbf{y}, \hat{\mathbf{u}}_{1}^{j-1} \mid1\right)}$ 
given that $W^{(j)}\left( \mathbf{y}, \hat{\mathbf{u}}_{1}^{j-1}| u_j\right)=\sum_{\mathbf{u}_{j+1}^{n}\in\{0,1\}^{N-i}}W(\mathbf{y|\mathbf{x}})$.
SC decoding remains attractive due to its low complexity and suitability for hardware implementations. It also serves as the foundation for more powerful variants such as the SCL decoder, which is our focus in the proposed protocol.


\section{The Proposed $(\epsilon,\delta)$-LDP Coding-Assisted Succinct Histogram Protocol}
\label{sec:protocol}
In this section, we first illustrate our proposed protocol. Next, we analyze the $(\epsilon,\delta)$-LDP guarantees, followed by the frequency estimation error analysis.
\subsection{The Proposed Protocol}
We consider $N$ clients with a single untrusted server that collects information from all clients for the unique heavy hitter problem in Definition~\ref{def:unique}. 
Client $i$ either holds the true item as a $k$-bit binary vector $\mathbf{u}^*\in\mathcal{U}=\{0,1\}^k$ or has no item as $\bot$, denoted by $\mathbf{u}_i\in\{\mathbf{u}^*\}\cup \{\bot\}$, for $i\in\{1,\dots,N\} \overset{\mathrm{def}}{=}[N]$.
We denote the set of clients that hold the unique item $\mathbf{u}^*$ as $\mathcal{T}$, so we have
\begin{equation}
    \mathbf{u}_i=\begin{cases}
        \mathbf{u}^* & \textnormal{if } i\in\mathcal{T}\\\bot & \textnormal{if } i\in[N]\setminus\mathcal{T}    \end{cases}
\end{equation}
for $i\in[N]$.
According to Definition~\ref{def:unique}, at least $\eta$ fraction of the clients have the unique item $\mathbf{u}^*$. 
Thus, we define the \textit{true frequency} of such item as $f(\mathbf{u}^*) \overset{\mathrm{def}}{=}\frac{|\mathcal{T}|}{N}\geq\eta$. 

Clients $i\in\mathcal{T}$ who hold the unique item $\mathbf{u}_i=\mathbf{u}^*$ encode $\mathbf{u}_i$ by the given $(n,k)$-polar code with a generator matrix $\mathbf{G}$ to obtain the corresponding codeword $\mathbf{v}_i=\mathbf{u}_i\mathbf{G}=\mathbf{u}^*\mathbf{G}$, then modulate the codeword by a normalized BPSK, which has the mapping $\{0,1\}^n\to\{-\frac{1}{\sqrt{n}},\frac{1}{\sqrt{n}}\}^n$, to generate $q(\mathbf{u}_i)=\mathbf{x}_i=\frac{1}{\sqrt{n}}(2\mathbf{v}_i-\mathbf{1}_n)=\frac{1}{\sqrt{n}}(2\mathbf{u}^*\mathbf{G}-\mathbf{1}_n)$;
for clients $i\in[N]\setminus\mathcal{T}$ who do not have the unique item, i.e. $\mathbf{u}_i=\bot$, set $q(\mathbf{u}_i)=\mathbf{x}_i=\mathbf{0}_n$. 
Thus, we summarize $q(\mathbf{u}_i)=\mathbf{x}_i$ as follows:
\begin{equation}
q(\mathbf{u}_i)=\mathbf{x}_i=
\begin{cases}\frac{1}{\sqrt{n}}(2\mathbf{u}^*\mathbf{G-\mathbf{1}_n}) & \textnormal{if }\mathbf{u}_i=\mathbf{u}^*\\\mathbf{0}_n& \textnormal{if }\mathbf{u}_i=\bot    \end{cases}
\end{equation}
for $i\in[N]$.
Next, client $i$ applies the randomized mechanism $\mathcal{R}_i$, that satisfies $(\epsilon,\delta)$-LDP, to obtain $\mathbf{z}_i=\mathcal{R}_i(\mathbf{u}_i) \overset{\mathrm{def}}{=} q(\mathbf{u}_i)+\mathbf{n}_i=\mathbf{x}_i+\mathbf{n}_i$, where $\mathbf{n}_i$ is a Gaussian noise vector that all entries are sampled from a zero mean Gaussian distribution $\mathcal{N}(0,\sigma^2)$, then sends $\mathbf{z}_i$ to the server, for $i\in[N]$. 
We leave the detailed setup of the randomized mechanism $\mathcal{R}_i$ to the next subsection.

The server collects $\mathbf{z}_i$'s from clients $i\in[N]$, then computes the mean $\mathbf{y}=\frac{1}{N}\sum_{i\in[N]}\mathbf{z}_i$. 
An SCL decoder is applied to decode the estimated item $\hat{\mathbf{u}}$ based on the LLR $\Tilde{\mathbf{y}} = \frac{2}{\sigma^2}\mathbf{y}$ such that $\hat{\mathbf{u}}=\mathrm{Dec}(\Tilde{\mathbf{y}})$. 
The frequency is estimated by calculating $\hat{f}(\hat{\mathbf{u}})=\langle q(\hat{\mathbf{u}}),\mathbf{y} \rangle=\langle \hat{\mathbf{x}},\mathbf{y}\rangle$, where $q(\hat{\mathbf{u}})=\hat{\mathbf{x}}=\frac{1}{\sqrt{n}}(2\hat{\mathbf{v}}-\mathbf{1}_n)=\frac{1}{\sqrt{n}}(2\hat{\mathbf{u}}\mathbf{G}-\mathbf{1}_n)$. Finally, the succinct histogram is generated as $\left(\hat{\mathbf{u}},\hat{f}(\hat{\mathbf{u}})\right)$, which includes the estimated item and the estimated frequency, respectively. 
We summarize the protocol in Fig.~\ref{fig:sys}.

\begin{figure}[t]
\centering
{\includegraphics[width=.5\textwidth]{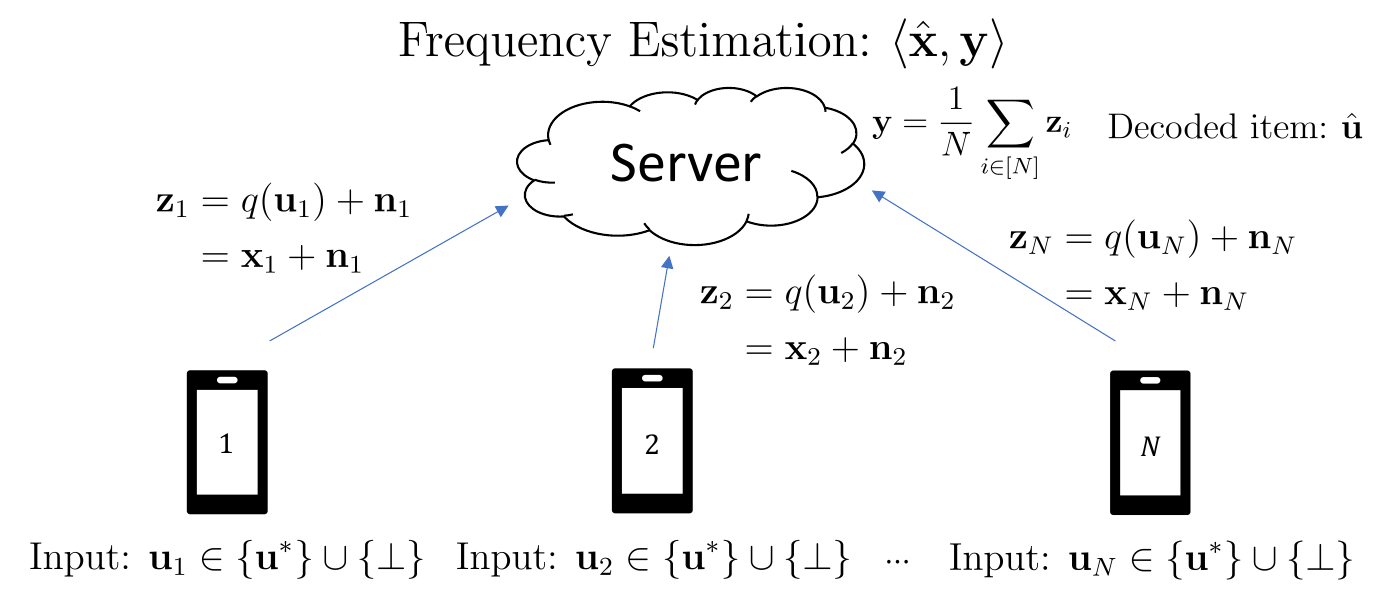}}\\
  \caption{Proposed Coding-Assisted Succinct Histogram Protocol}\label{fig:sys}
\end{figure}

\subsection{The $(\epsilon,\delta)$-LDP guarantee}
To ensure $(\epsilon,\delta)$-LDP by adding perturbations sampled from a zero mean Gaussian distribution $\mathcal{N}(0,\sigma^2)$, we must identify the variance $\sigma^2$ of such distributions. 
We start by analyzing the sensitivity of the proposed protocol as follows.

\begin{lemma}[Sensitivity]
    The sensitivity of our proposed protocol is bounded as follows:
    \begin{equation}
        \Delta=\max_{\mathbf{u^*},\hat{\mathbf{u}}}\, \lVert q(\mathbf{u}^*) - q(\hat{\mathbf{u}}) \rVert_2 = 2.
    \end{equation}
\end{lemma}
\begin{proof}
    Given that $q(\mathbf{u}^*)=\mathbf{x}^*=\frac{1}{\sqrt{n}}\left(2\mathbf{u}^*\mathbf{G}-\mathbf{1}_n\right)$ and $q(\hat{\mathbf{u}})=\hat{\mathbf{x}}=\frac{1}{\sqrt{n}}\left(2\hat{\mathbf{u}}\mathbf{G}-\mathbf{1}_n\right)$, we have
    \begin{equation}
    \label{eq:sensitivity}
    \begin{aligned}
        \left\lVert q(\mathbf{u}^*)-q(\hat{\mathbf{u}})\right\rVert_2
        &=\left\lVert\frac{1}{\sqrt{n}}(2\mathbf{u}^*\mathbf{G}-\mathbf{1}_n)-\frac{1}{\sqrt{n}}(2\hat{\mathbf{u}}\mathbf{G}-\mathbf{1}_n)\right\rVert_2\\
        &=\left\lVert \frac{2}{\sqrt{n}}(\mathbf{u}^*-\hat{\mathbf{u}})\mathbf{G} \right\rVert_2\\
        &\leq \left\lVert \frac{2}{\sqrt{n}}(\mathbf{1}_k-\mathbf{0}_k)\mathbf{G} \right\rVert_2\\
        &=\left\lVert \frac{2}{\sqrt{n}}\cdot\mathbf{1}_n \right\rVert_2=2.
    \end{aligned}
    \end{equation}
    Thus, the maximum value of $\lVert q(\mathbf{u}^*) - q(\hat{\mathbf{u}}) \rVert_2$ given any $\mathbf{u}^*,\hat{\mathbf{u}}\in\mathcal{U}$, which is the sensitivity, is $\Delta=2$. Also, this maximum occurs since the all-one vector is a codeword in the polar codebook.
\end{proof}

\begin{remark}
    We can reduce the sensitivity of our protocol, as stated in in \eqref{eq:sensitivity}, by carefully adjusting the code structure. This can be done by simply setting some of the information bits in the given $(n,k)$-polar code to zeros. 
    For instance, one can exclude the all-one codeword by setting the information bit corresponding to the last index to zero. 
    Thus, the sensitivity becomes $\Delta=\left\lVert \frac{2}{\sqrt{n}}\cdot\mathbf{c}_{\textnormal{max-weight} }\right\rVert_2<2$, where $\mathbf{c}_{\textnormal{max-weight} }$ is the codeword with the maximum weight excluding the all-one codeword $\mathbf{1}_n$. This gives rise to an interesting problem, i.e., minimizing the maximum weight of codewords, for polar codes which we leave for future work.
\end{remark}

With the sensitivity $\Delta$, to ensure $(\epsilon,\delta)$-LDP of our protocol, we employ the Analytic Gaussian mechanism \cite{balle2018improving} that achieves the optimal noise variance. This leads to the following theorem with a proof that follows from the $(\epsilon,\delta)$-LDP guarantee of the underlying analytic Gaussian mechanism. 


\begin{theorem}[$(\epsilon,\delta)$-LDP for the proposed protocol]
\label{thm:agm_CASH}
The protocol is $(\epsilon,\delta)$-LDP with the analytic Gaussian mechanism as the randomized mechanism $\mathcal{R}_i$, such that
$\mathcal{R}_i(\mathbf{u}_i) =q(\mathbf{u}_i)+\mathbf{n}_i= \mathbf{x}_i+\mathbf{n}_i$,
where $\mathbf{n}_i$ is a noise vector sampled from $\mathcal{N}(0,\sigma^2)$, which is a zero-mean Gaussian distribution with variance $\sigma^2=\textsc{AnalyticGaussian}(\epsilon,\delta,\Delta)$, for $i\in[N]$.
\end{theorem}

With Theorem~\ref{thm:agm_CASH}, we summarize our protocol in Algorithm~\ref{alg:LDPCASH}.


\begin{algorithm}[t]
    \caption{$(\epsilon,\delta)$-LDP Coding-Assisted Succinct Histogram Protocol for the Unique Heavy Hitter Problem}\label{alg:LDPCASH}
    \begin{algorithmic}[1]
    \renewcommand{\algorithmicrequire}{\textbf{Input:}}
    \renewcommand{\algorithmicensure}{\textbf{Output:}}
    \REQUIRE Clients' inputs $\{\mathbf{u}_i\in\mathbf{u}^*\cup\bot:i\in[N]\}$, privacy guarantee $\epsilon$ and $\delta$, sensitivity $\Delta$, and generator matrix $\mathbf{G}$ constructed by an $(n,k)$-polar code
    \ENSURE Succinct histogram $\left(\hat{\mathbf{u}},\hat{f}(\hat{\mathbf{u}})\right)$
    
    \FOR{Client $i\in[N]$}
        \STATE If $\mathbf{u}_i\neq\bot$, client $i$ encodes then modulates its item $\mathbf{x}_i=\frac{1}{\sqrt{n}}(2\mathbf{v}_i-\mathbf{1}_n)=\frac{1}{\sqrt{n}}(2\mathbf{u}_i\mathbf{G}-\mathbf{1}_n)$. Else, $\mathbf{x}_i=\mathbf{0}_n$.
        \STATE Client $i$ randomized its private report $\mathbf{z}_i=\mathcal{R}_i\left(\mathbf{x}_i\right)=\mathbf{x}_i+\mathbf{n}_i$, where $\mathbf{n}_i$ is a vector that all entries are sampled from a zero mean Gaussian distribution $\mathcal{N}(0,\sigma^2)$ with $\sigma^2=\mathrm{AnalyticGaussian}(\epsilon,\delta,\Delta).$
        \STATE Client $i$ sends $\mathbf{z}_i$ to the server.
    \ENDFOR
    \STATE Server collects $\mathbf{z}_i$'s from clients $i\in[N]$.
    \STATE Server computes $\mathbf{y}=\frac{1}{N}\sum_{i\in[N]}\mathbf{z}_i$.
    \STATE Server calculates the LLR $\Tilde{\mathbf{y}}=\frac{2}{\sigma^2}\mathbf{y}$.
    \STATE Server decodes $\Tilde{\mathbf{y}}$ into an estimate for the unique item $\hat{\mathbf{u}}=\mathrm{Dec}(\Tilde{\mathbf{y}})$ by an SCL decoder.
    \STATE Server encodes then modulates its decoded item $\hat{\mathbf{x}}=\frac{1}{\sqrt{n}}(2\hat{\mathbf{v}}-\mathbf{1}_n)=\frac{1}{\sqrt{n}}(2\hat{\mathbf{u}}\mathbf{G}-\mathbf{1}_n)$.
    \STATE Server computes a frequency estimate $\hat{f}(\hat{\mathbf{u}})=\langle \hat{\mathbf{x}},\mathbf{y}\rangle$.
    \STATE \textbf{Return} the succinct histogram $\left(\hat{\mathbf{u}},\hat{f}(\hat{\mathbf{u}})\right)$
    \end{algorithmic}
\end{algorithm}
\begin{remark}
The major differences between our proposed protocol Algorithm~\ref{alg:LDPCASH} compared to the most relevant prior work \prot in \cite[Algorithm 2]{bassily2015local} are as follows: 
(1) We formulate the practical code construction for encoding the unique item to the codewords in step 2, which \prot did not specify the construction at all;
(2) The randomizers between our protocol and \prot are fundamentally different. \prot randomly picks the $j$th entry from $\mathbf{x}_i$, which is denoted by $x_{i,j}$, randomizes it to either $c_\epsilon n x_{i,j}$ or $-c_\epsilon n x_{i,j}$ according to the randomized response \cite{Warner01031965} when $\mathbf{x}_i\neq\mathbf{0}_n$, or assigning $c_\epsilon\sqrt{n}$ or $-c_\epsilon\sqrt{n}$ uniformly at random, while other entries are set to zeros. Our proposed protocol randomizes each entry by adding Gaussian noise.
(3) Our protocol enables flexibility in the soft-decision process by adding perturbations to the queries by the analytic Gaussian mechanism in step 3, that satisfy $(\epsilon,\delta)$-LDP with optimal noise. In \protns, a local randomizer inspired by randomized response was applied to perturb the queries, making it fundamentally binarized to a hard-decision process;
(4) The SCL decoder used in step 9 is a soft decoder, while \prot rounds $\mathbf{y}$ in our step 7 to $(-\frac{1}{\sqrt{n}},\frac{1}{\sqrt{n}})^n$, which is limited to a hard decoder. Section~\ref{sec:sim} demonstrates the significance of the decoder, especially when it comes to a low true frequency.
\end{remark}

\subsection{Error Analysis}
We analyze the error in the frequency estimation of the correctly decoded item, i.e., $\hat{\mathbf{u}}=\mathbf{u}^*$. The following lemma derives the expression of $\mathbf{y}$, which is the mean of collected $\mathbf{z}_i$'s from clients $i\in[N]$.
\begin{lemma}
\label{lem:y}
 The received vector at the server in step 7 in Algorithm~\ref{alg:LDPCASH} is 
 \begin{equation}
     \mathbf{y}=\frac{1}{N}\sum_{i\in[N]}\mathbf{z}_i=f(\mathbf{u}^*)\mathbf{x}^*+\frac{1}{N}\sum_{i\in[N]}\mathbf{n}_i,
 \end{equation}
\end{lemma}
\begin{proof}
Given that $\mathbf{z}_i=\mathbf{x}_i+\mathbf{n}_i$ for clients $i\in[N]$, we have 
    \begin{equation}
    \label{eq:y_sub1}
        \begin{aligned}
            \mathbf{y}&=\frac{1}{N}\sum_{i\in[N]}\mathbf{z}_i=\frac{1}{N}\sum_{i\in[N]}\mathbf{x}_i+\frac{1}{N}\sum_{i\in[N]}\mathbf{n}_i\\
        \end{aligned}.
    \end{equation}
Then, note that $\mathbf{x}_i=\mathbf{x}^*$ for $i\in\mathcal{T}$ and $\mathbf{x}_i=\mathbf{0}_i$ for $i\in[N]\setminus\mathcal{T}$.
    \begin{equation}
    \label{eq:y_sub2}
        \begin{aligned}
            \frac{1}{N}\sum_{i\in[N]}\mathbf{x}_i
            &=\frac{1}{N}\left[\sum_{i\in\mathcal{T}}\mathbf{x}_i+\sum_{i\in[N]\setminus\mathcal{T}}\mathbf{x}_i\right]\\
            &=\frac{1}{N}\left[|\mathcal{T}|\cdot\mathbf{x}^*+\left(N-|\mathcal{T}|\right)\cdot\mathbf{0}_n\right]\\
            &=\frac{|\mathcal{T}|}{N}\mathbf{x}^*= f(\mathbf{u}^*)\mathbf{x}^*.
        \end{aligned}
    \end{equation}
    Combining \eqref{eq:y_sub1} and \eqref{eq:y_sub2} proves the statement.
\end{proof}


\begin{lemma}[Error analysis for the correctly decoded items]
\label{lem:Error Bound}
    Assuming the unique item is correctly decoded, i.e., $\hat{\mathbf{u}}=\mathbf{u}^*$, the frequency estimation error is expressed as follows:
    \begin{equation}
        \mathrm{Err}_{\hat{\mathbf{u}},\mathbf{u}^*}\overset{\mathrm{def}}{=}\left|\hat{f}(\hat{\mathbf{u}})-f(\mathbf{u^*})\right|=\left|\langle\mathbf{x}^*,\frac{1}{N}\sum_{i\in[N]}\mathbf{n}_i\rangle\right| .
    \end{equation}
\end{lemma}
\begin{proof}
When we have the correct decoded item, we can substitute $\hat{\mathbf{u}}$ by $\mathbf{u}^*$, thus the problem becomes $\mathrm{Err}_{\hat{\mathbf{u}}=\mathbf{u}^*}=\left|\hat{f}({\mathbf{u}}^*)-f(\mathbf{u^*})\right|$. 
Given the frequency estimate $\hat{f}(\hat{\mathbf{u}})=\langle\hat{\mathbf{x}},\mathbf{y}\rangle$ in the protocol, plugging $\hat{\mathbf{u}}=\mathbf{u}^*$, we obtain that $\hat{f}({\mathbf{u}}^*)=\langle{\mathbf{x}^*},\mathbf{y}\rangle$, where $\hat{\mathbf{x}}=\mathbf{x}^*$, yields $\mathrm{Err}_{\hat{\mathbf{u}},\mathbf{u}^*}=\left|\langle{\mathbf{x}^*},\mathbf{y}\rangle-f(\mathbf{u}^*)\right|$. Then, with Lemma~\ref{lem:y} we have
\begin{equation}
\begin{aligned}   \mathrm{Err}_{\hat{\mathbf{u}}=\mathbf{u}^*} =&\left|\langle{\mathbf{x}^*},f(\mathbf{u}^*)\mathbf{x}^*+\frac{1}{N}\sum_{i\in[N]}\mathbf{n}_i\rangle-f(\mathbf{u}^*)\right|\\
    =&\left|f(\mathbf{u}^*)\langle{\mathbf{x}^*},\mathbf{x}^*\rangle+\langle{\mathbf{x}^*},\frac{1}{N}\sum_{i\in[N]}\mathbf{n}_i\rangle-f(\mathbf{u}^*)\right|\\
    =&\left|\langle\mathbf{x}^*,\frac{1}{N}\sum_{i\in[N]}\mathbf{n}_i\rangle\right|,
\end{aligned}
\end{equation}
where $\langle\mathbf{x}^*,\mathbf{x}^*\rangle=1$.
\end{proof}




\section{Simulation Results}
\label{sec:sim}
In this section, we analyze how true frequencies, codelength, and number of clients, affect the performance. 
We consider a $(64,8)$-polar code, which implies there are $2^{8}$ possible unique items, with the number of clients $N=1000$, and we analyze the true frequencies $f(\mathbf{u}^*) \in \{0.5,0.6,0.7\}$. The simulations vary $\epsilon$ and set $\delta=10^{-4}$.
To decode the estimated unique item, we consider the SC list (SCL) decoder \cite{tal2015list} with a list size $L=8$ for our protocol and the maximum likelihood decoder for \prot in \cite{bassily2015local}. 
The BLER and frequency estimation error comparisons to $\epsilon$ are demonstrated in Figs.~\ref{fig:varf_bler} and \ref{fig:varf_fqe}, respectively.

Figure~\ref{fig:varf_bler} demonstrates the robustness of the protocols that experience different true frequencies. 
When the true frequency is $f(\mathbf{u}^*)=0.7$, both protocols have good results, but \prot achieves a lower BLER due to the high performance of the maximum likelihood decoder. 
When the true frequency becomes $f(\mathbf{u}^*)=0.6$, our protocol still perform well, while \prot starts to suffer an error floor since the protocol rounds the entries in $\mathbf{y}$ to $\frac{1}{\sqrt{n}}$, if $y_j\geq0$ and to $-\frac{1}{\sqrt{n}}$, if $y_j<0$, for $j\in[n]$.
Intuitively, since lower true frequency $f(\mathbf{u}^*)$ means more clients have $\mathbf{x}_i=\mathbf{0}_n$, more entries tend to be rounded to $\frac{1}{\sqrt{n}}$, which gradually causes the entries in $\Tilde{\mathbf{y}}$ biased towards positive. 
This becomes worse when the true frequency is decreased to $f(\mathbf{u}^*)=0.5$, where most items are decoded incorrectly.
In Fig.~\ref{fig:varf_fqe}, one can observe that our frequency estimation outperforms \protns. 
These results indicate that our proposed protocol has versatile performance regarding the BLER given a wide range of true frequencies, equivalent to having a high correct decoding rate. 
Thus, a lower $\eta$ can be guaranteed such that our protocol is better than \prot according to Definition~\ref{def:unique}, which requires $\eta$ as small as possible.
Furthermore, the proposed protocol has a lower complexity of the decoder while achieving better frequency estimation errors.

In Fig.~\ref{fig:varNn_bler}, we evaluate the performance of our proposed protocol in terms of BLER for varying the codelength from $n=64$ to $n=256$ and the number of clients from $N=500$ to $N=1000$, given $\delta=10^{-4}$, $f(\mathbf{u}^*)=0.7$, and $k=8$. We can observe that increasing the codelength from $n=64$ to $n=256$ gives our protocol a performance gain for both $N=500$ and $N=1000$. Furthermore,  increasing the number of clients from $N=500$ to $N=1000$ gives us a performance gain, which can be observed from Lemma~\ref{lem:Error Bound}.

\begin{figure}[t]
\centering
{\includegraphics[width=.49\textwidth]{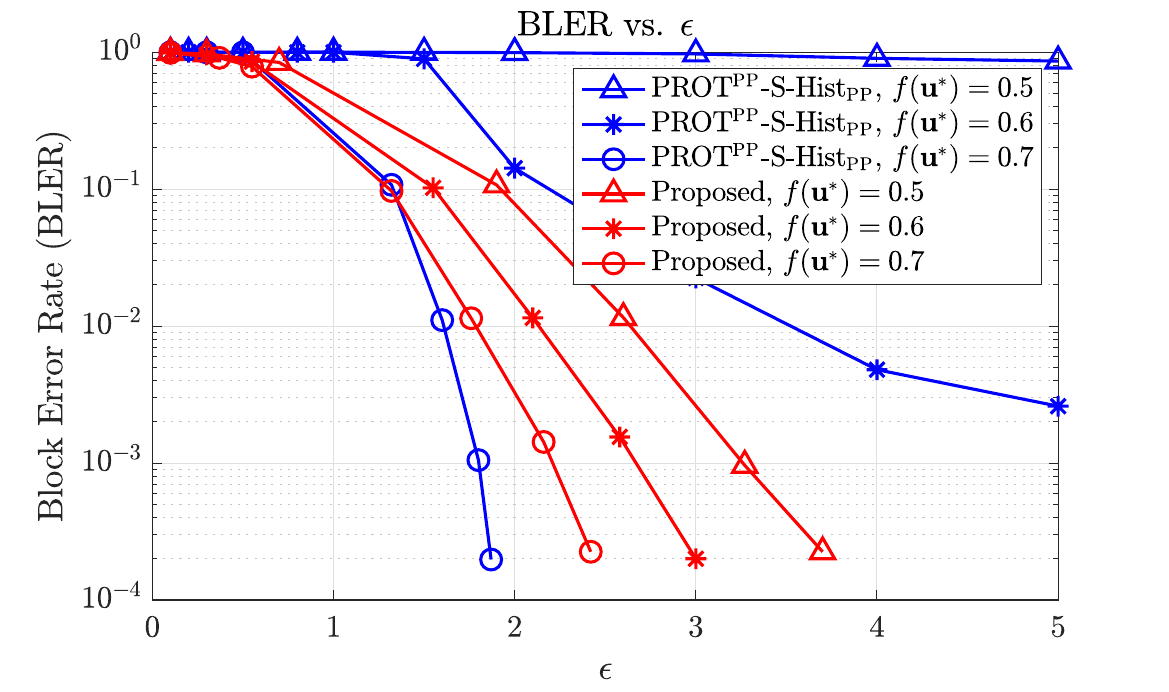}}\\
  \caption{BLER vs. $\epsilon$ for $(n,k)=(64,8)$ and $N=1000$}\label{fig:varf_bler}
\end{figure}

\begin{figure}[t]
\centering
{\includegraphics[width=.49\textwidth]{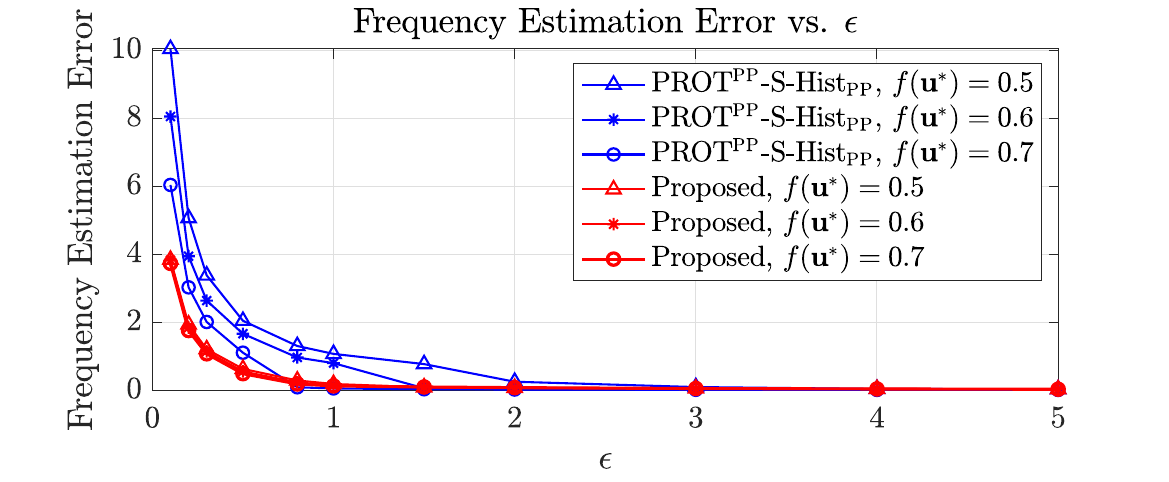}}\\
  \caption{Frequency Estimation Error vs. $\epsilon$ for $(n,k)=(64,8)$ and $N=1000$}\label{fig:varf_fqe}
\end{figure}

\begin{figure}[t]
\centering
{\includegraphics[width=.49\textwidth]{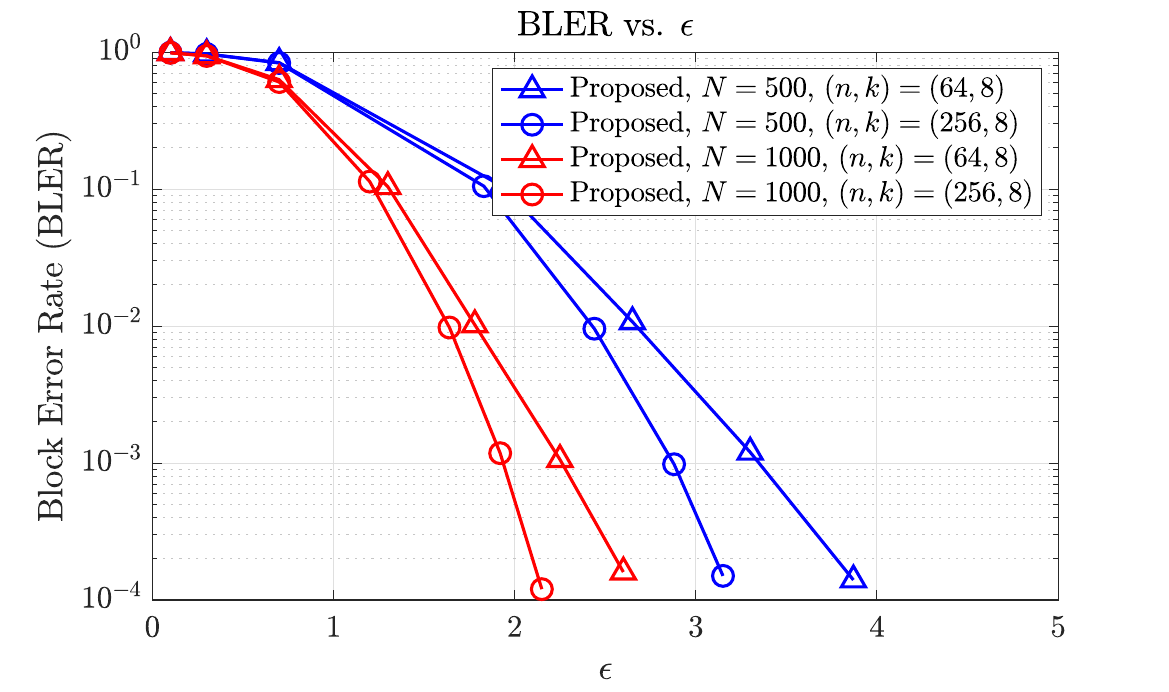}}\\
  \caption{BLER vs. $\epsilon$ for $f(\mathbf{u}^*)=0.7$ }\label{fig:varNn_bler}
\end{figure}


\section{Conclusion}
\label{sec:con}

In this work, we introduced the first practically implementable protocol for constructing succinct histograms under $(\epsilon,\delta)$-LDP using error-correcting codes. By leveraging polar codes and an SCL decoder, our design supports soft decoding through an analytic Gaussian mechanism, enabling improved robustness without sacrificing frequency estimation accuracy. Unlike prior works that remained theoretical, we provide concrete constructions and experimental evaluations, demonstrating that our protocol outperforms across varied settings, while achieving similar frequency estimation errors, especially on lower true frequencies of the unique items.

\bibliographystyle{IEEEtran}
\bibliography{Paper}


\end{document}